\documentclass[]{article}
\usepackage[utf8]{inputenc}
\usepackage[pdftex]{graphicx}
\usepackage{amsmath,amsfonts,amsbsy,amssymb,amsthm}
\usepackage{epstopdf}
\usepackage{tcolorbox}
\usepackage{mathtools}
\usepackage{physics}
\usepackage[caption=false]{subfig}
\usepackage{multirow}
  \usepackage{geometry}

    \newgeometry{vmargin={20mm}, hmargin={20mm,20mm}} 
\usepackage{pifont}
\usepackage{amssymb,amsmath}

\usepackage{array}
\newcolumntype{P}[1]{>{\centering\arraybackslash}p{#1}}
\newcolumntype{M}[1]{>{\centering\arraybackslash}m{#1}}

\usepackage{diagbox}
\newcommand{\bigzero}{\mbox{\normalfont\Large\bfseries 0}}
\newcommand{\rvline}{\hspace*{-\arraycolsep}\vline\hspace*{-\arraycolsep}}
\newcommand{\bbsout}[1]{}

\newcommand{\bsout}[1]{}
\usepackage{hyperref}
\usepackage{pdfpages}
\usepackage[T1]{fontenc}

\usepackage{amsthm}

\newtheorem{theorem}{Theorem}

\newtheorem{corollary}[theorem]{Corollary}

\def \cH{{\cal H}}

\def\ANU{Centre for Quantum Computation and Communication Technology, Department of Quantum Science, Australian National University, Canberra, ACT 2601, Australia.}

\usepackage{authblk}
\def\UEC{Graduate School of Informatics and Engineering, The
  University of Electro-Communications, Tokyo 182-8585, Japan}
  \def\Astarnew{A*STAR Quantum Innovation Centre (Q.InC), Institute of Materials Research and Engineering (IMRE), Agency for Science Technology and Research (A*STAR), 2 Fusionopolis Way, 08-03 Innovis 138634, Singapore}
   
\begin{document}

\newlength\figHeight 
\newlength\figWidth 

\title{Role of the extended Hilbert space in the attainability of the Quantum Cram\'{e}r--Rao bound for multiparameter estimation}

\author[1]{Lorc\'{a}n O. Conlon\footnote{lorcanconlon@gmail.com}}
\author[2]{Jun Suzuki\footnote{junsuzuki@uec.ac.jp}}
\author[1,3]{Ping Koy Lam}
\author[1,3]{Syed M. Assad\footnote{cqtsma@gmail.com}}
\affil[1]{\Astarnew}
\affil[2]{\UEC}
\affil[3]{\ANU}
\maketitle

\begin{abstract}
The symmetric logarithmic derivative Cram\'{e}r--Rao bound (SLDCRB) provides a fundamental limit to the minimum variance with which a set of unknown parameters can be estimated in an unbiased manner. It is known that the SLDCRB can be saturated provided the optimal measurements for the individual parameters commute with one another. However, when this is not the case the SLDCRB cannot be attained in general. In the experimentally relevant setting, where quantum states are measured individually, necessary and sufficient conditions for when the SLDCRB can be saturated are not known. In this setting the SLDCRB is attainable provided the SLD operators can be chosen to commute on an extended Hilbert space. However, beyond this relatively little is known about when the SLD operators can be chosen in this manner. In this paper we present explicit examples which demonstrate novel aspects of this condition. Our examples demonstrate that the SLD operators commuting on any two of the following three spaces: support space, support-kernel space and kernel space, is neither a necessary nor sufficient condition for commutativity on the extended space. We  present a simple analytic example showing that the Nagaoka--Hayashi Cramér-Rao bound is not always attainable. Finally, we provide necessary and sufficient conditions for the attainability of the SLDCRB in the case when the kernel space is one-dimensional. These results provide new information on the necessary and sufficient conditions for the attainability of the SLDCRB.
\end{abstract}

\section{Introduction}
The estimation of unknown parameters is a fundamental part of doing science. As such, there is great excitement surrounding the possibility that quantum resources may offer an advantage in this area~\cite{aasi2013enhanced,casacio2021quantum,leibfried2004toward,marciniak2021optimal,boixo2008quantum,caves1981quantum,guo2020distributed}. However, much of the previous work on this area has focused on single parameter estimation. From a practical perspective this is not ideal, as many physical problems are intrinsically multiparameter problems~\cite{monras2011measurement,humphreys2013quantum,baumgratz2016quantum,cimini2019quantum,vidrighin2014joint,chrostowski2017super,szczykulska2017reaching,hou2020minimal,conlon2023multiparameter,conlon2023verifying}. Additionally, the incompatibility of multiple observables, a fundamentally quantum feature, is only relevant when considering the simultaneous estimation of multiple parameters~\cite{heisenberg1985anschaulichen,robertson1929uncertainty,arthurs1965simultaneous}. For recent reviews on the topic see Refs.~\cite{szczykulska2016multi,liu2019quantum,sidhu2020geometric,demkowicz2020multi,conlon2023quantum}.
%
%
%

Due to the inherently probabilistic nature of quantum mechanics, when we perform any measurement on a system there will be some fundamental unavoidable uncertainty associated with our estimate of the unknown parameters. It is an important task to design the experiment, and subsequent estimation of the unknown parameters, in such a way as to minimise this uncertainty. However, this is a difficult task which quickly becomes numerically intractable for estimating even a small number of parameters on low-dimensional systems. To get around this difficulty, a series of quantum Cram{\'{e}}r--Rao bounds have been developed~\cite{helstrom1967minimum,helstrom1968minimum,yuen1973,gill2005state,holevo1973statistical,holevo2011probabilistic,nagaoka2005new,nagaoka2005generalization,conlon2021efficient}, which provide lower bounds on how small the mean squared error (MSE) in the estimate of the unknown parameters can be. The different Cram{\'{e}}r--Rao bounds apply in different settings. For example, the Holevo Cram{\'{e}}r--Rao bound, $\mathcal{C}_\text{H}$, (HCRB)~\cite{holevo1973statistical,holevo2011probabilistic} can be asymptotically approached through collective measurements on infinitely many copies of the quantum state~\cite{kahn2009local, yamagata2013quantum, yang2019attaining}. The Nagaoka--Hayashi Cram{\'{e}}r--Rao bound, $\mathcal{C}_\text{NH}$, (NHCRB)~\cite{nagaoka2005new,nagaoka2005generalization,conlon2021efficient} applies when we consider performing measurements on a single copy of the quantum state. However, it has recently been shown that this is not always a tight bound in this setting~\cite{hayashi2023tight}. The symmetric logarithmic derivative Cram{\'{e}}r--Rao bound, $\mathcal{C}_\text{S}$, (SLDCRB)~\cite{helstrom1967minimum,helstrom1968minimum} (note that this is often referred to as the quantum Cram{\'{e}}r--Rao bound or Helstrom bound) is of particular interest as when estimating a single parameter it represents a tight bound, i.e. there is always a measurement which saturates this bound~\cite{braunstein1994statistical}.

The necessary and sufficient conditions for saturating the SLDCRB in the limit where one can perform a collective measurement on infinitely many copies of the quantum state are known~\cite{ragy2016compatibility}. The attainability of the SLDCRB in the single-copy setting is a more subtle issue and has been listed as one of five open problems in quantum information theory~\cite{horodecki2022five}. Partial solutions to this problem have been presented in the view of quantum metrology~\cite{pezze2017optimal,yang2019optimal,PhysRevLett.128.250502,PhysRevA.105.062442,amari2000methods,nurdin2024saturability}. For full rank probe states the problem is completely solved. A necessary and sufficient condition for saturating the SLDCRB in this case is that the SLD operators, $L_i$, must commute with one another, $[L_i,L_j]=0,\;\forall i,j$. Furthermore, Suzuki, Yang and Hayashi (see Appendix B1 of Ref.~\cite{suzuki2020quantum}) gave the necessary and sufficient condition for the most general rank-deficient case. In this sense, the open problem appears to be solved. An important point for understanding the solution presented in Ref.~\cite{suzuki2020quantum} is the fact that in the rank-deficient case, the SLD operators are not unique - we are free to choose certain elements of these operators. Suzuki, Yang and Hayashi's solution then states that the SLDCRB is attainable if the SLD operators can be chosen such that they commute on the extended space. However, their solution does not provide conditions for when the SLD operators can be chosen in this manner. In this sense, the open problem in Ref.~\cite{horodecki2022five} could be considered not solved~\cite{Lukaszcomms2023}. The extra degrees of freedom associated with the elements of the SLD operators we are free to choose makes the problem significantly more general and harder to solve. Finally, we note that it has recently been shown that if the SLDCRB cannot be saturated in the single-copy setting then it cannot be saturated with a collective measurement on any finite number of copies of the quantum state~\cite{conlon2022gap}\textemdash~providing further motivation for solving the single-copy attainability problem.

In this paper we show the importance of the extended space when considering this problem. We present five examples which demonstrate different aspects of the necessary and sufficient conditions for saturating the SLDCRB. Our examples demonstrate the following: 1) The SLD operators commuting on the support of the probe state is a necessary but not sufficient condition for the attainability of the SLDCRB. 2) The SLD operators commuting on the support and support-kernel space is neither a necessary nor sufficient condition for the attainability of the SLDCRB. 3) The SLD operators commuting on the support and kernel space is neither a necessary nor sufficient condition for the attainability of the SLDCRB. We then provide necessary and sufficient conditions for the attainability of the SLDCRB in the case when the kernel space is one-dimensional. The rest of the paper is set out as follows. We first introduce the necessary preliminary material in section~\ref{sec:prelim}. In section~\ref{sec:equality} we present a brief summary of the known conditions for various Cram{\'{e}}r--Rao bounds to be equal. Then in sections~\ref{egA} to \ref{exE} we present our examples. In section.~\ref{sec:1D}, we present our results for the one-dimensional kernel case. Finally, in section~\ref{sec:conc} we conclude. 

%
%

\section{Preliminaries}
\label{sec:prelim}
We consider a density matrix $\mathcal{S}_\theta$ in a $d$ dimensional Hilbert space $\mathcal{H}$, which is a function of $n$ parameters $\mathcal{S}_\theta$, $\theta=(\theta_1,\theta_2,\hdots,\theta_n)$\footnote{Going forward, we shall drop the explicit dependence on $\theta$.}. We shall consider rank-deficient quantum states with rank $r<d$. The quantum state has an associated support and kernel. We shall write the density matrix $\mathcal{S}$ as
\begin{equation}
\mathcal{S}=\sum_{i=1}^{r}c_i\ket{s_i}\bra{s_i}\;,
\end{equation}
where $\sum_{i=1}^{r}c_i=1$, $c_i>0$, and we denote a set of eigenvectors for the support space by $\{\ket{s_i}\}$. The kernel space is also defined by $\mathcal{S}$ and its derivatives at the true parameter value. We shall denote the orthonormal basis vectors for the kernel space as $\{\ket{k_i}\}$. The derivatives of a quantum state may exist in the support and support-kernel spaces but not in the kernel space, where the support kernel space is spanned by the outer product of the support and kernel vectors ${\ket{s_i}\bra{k_j}}$ and ${\ket{k_i}\bra{s_j}}$. We can therefore choose a basis such that the quantum state and its derivatives can be written as 
\begin{align}
\label{Sdiagbasis}
\mathcal{S}
  = \begin{pmatrix}S^\text{s}&\bigzero\\
    \bigzero&\bigzero\end{pmatrix}\;,\;
S_j
  = \begin{pmatrix}S_j^\text{s}&S_j^\text{sk}\\[0.2cm]
    S_j^\text{ks}&\bigzero\end{pmatrix}\;,
\end{align}
where we use $\mathbf{0}$ to denote a block matrix of zeros. Note that $\mathcal{S}$ existing in $\mathcal{H}$, does not necessarily mean that the derivatives span the whole space. Throughout this manuscript we shall assume that the $S_j$ are linearly independent to ensure the problem of estimating the unknown parameters is well defined. 

\subsection{Most informative bound}
\label{secCMI}
We can define the most informative bound $\mathcal{C}_\text{MI}$ as the MSE which would be attained from the best possible measurement on individual quantum states. Therefore, the conditions for the attainability of the SLDCRB refers to the conditions for $\mathcal{C}_\text{S}=\mathcal{C}_\text{MI}$. The most informative bound can be defined in terms of a positive operator valued measure (POVM). A POVM is described by a set of positive linear
operators, $\{\Pi_{k}\}$, which sum to the identity
\begin{equation}
\label{Eq:BGsumpi}
\sum_{k}\Pi_{k}=\mathbb{I}_d\;.
\end{equation}
The $k$-th measurement outcome occurs with probability
$p_k=\text{tr}\{\mathcal{S}\Pi_{k}\}$. Based on the different measurement outcomes we can
construct an unbiased estimator for the parameters of interest,
$\hat{\theta}$. An estimator maps each measurement outcome to an estimate of $\theta$.
Let $\hat{\theta}_{j,k}$ be the estimate of $\theta_j$ corresponding to the measurement outcome $k$. A locally unbiased estimator must satisfy
\begin{equation}
\sum_k\hat{\theta}_{j,k}\text{tr}\{\mathcal{S}\Pi_{k}\}=\theta_{j}\;,
\label{Eq:Bg:POVMunbiased1}
\end{equation}
where $\theta_{j}$ is the local parameter value and
\begin{equation}
\sum_k\hat{\theta}_{j,k}\text{tr}\{S_i\Pi_{k}\}=\delta_{i,j}\;.
\label{Eq:Bg:POVMunbiased2}
\end{equation}
 For estimating several parameters
simultaneously the MSE matrix,
$V$, corresponding to a given POVM and estimator coefficients has elements given by
\begin{equation}
\label{eq:parmsemat}
[V]_{ij}=\sum_{k}(\hat{\theta}_{i,k}-\theta_{i})(\hat{\theta}_{j,k}-\theta_{j})p_k\;,
\end{equation}
where the sum is over all possible
measurement outcomes. We shall denote the optimal POVM and estimator coefficients which minimise the trace of $V$ as $\{\Pi_{k}^*\}$ and $\hat{\theta}_{j,k}^*$. This POVM and estimator coefficients saturate the most informative bound.

\subsection{Symmetric logarithmic derivative Cram\'{e}r--Rao bound}

The SLD operators $L_i$ are defined through the following equation
\begin{equation}
\label{eq:SLDdef}
S_i=\frac{1}{2}(\mathcal{S}L_i+L_i\mathcal{S})\;.
\end{equation}
Note that the elements of the SLD operators in the kernel space are not specified by this equation. When we do not consider an extended space, the SLD operator can be written in the same basis used in Eq.~\eqref{Sdiagbasis} as
\begin{align}
 L_j
  = \begin{pmatrix}L_j^\text{s}&L_j^\text{sk}\\
    L_j^\text{ks}&L_j^\text{k}\\
 \end{pmatrix}\;.
\end{align}
From Eq.~\eqref{eq:SLDdef}, the $L_j^\text{s}$ and $L_j^\text{sk}$ elements are uniquely defined, whereas the $L_j^\text{k}$ can be chosen freely. The quantum Fisher information matrix is then defined as
\begin{equation}
J_{\text{S},ij}=\frac{1}{2}\text{tr}\{\mathcal{S}(L_iL_j+L_jL_i)\}\;,
\end{equation}
and the corresponding SLDCRB is
\begin{equation}
\mathcal{C}_\text{S}=\text{Tr}\{J_{\text{S}}^{-1}\}\;,
\end{equation}
where we use $\text{Tr}\{\}$ to denote the trace over a classical matrix. The elements of the SLD operators in the kernel space do not influence the SLDCRB, however they can in some instances be chosen such that the SLD operators either commute or do not commute.

 Given a rank-deficient state $\mathcal{S}$ and SLD operators $L_j$,
we can extend the state and its derivatives by introducing an
extended space $\tilde{H}=H\oplus H_{\text{e}}$ where $H_{\text{e}}$ has dimensions $e$. Therefore, the extended Hilbert space has dimensions $d+e$. In this extended space, the quantum state and its derivatives can be written as
\begin{align}
\mathcal{S}
  = \begin{pmatrix}S^\text{s}&\bigzero&\bigzero\\
    \bigzero&\bigzero&\bigzero\\
    \bigzero&\bigzero&\bigzero\end{pmatrix}\;,\;
S_j
  = \begin{pmatrix}S_j^\text{s}&S_j^\text{sk}&\bigzero\\[0.2cm]
    S_j^\text{ks}&\bigzero&\bigzero\\[0.2cm]
    \bigzero&\bigzero&\bigzero\end{pmatrix}\;.
\end{align}
Note that this does not correspond to any physical change in the problem. Similarly, as the SLD operators are defined only in the support and support-kernel
space, we are free to extended them to
\begin{align}
\label{eqSLDopextended}
  \mathcal{L}_j
  = \begin{pmatrix}L_j^\text{s}&L_j^\text{sk}&\bigzero\\[0.2cm]
    L_j^\text{ks}&L_j^\text{k}&L_j^\text{ke}\\[0.2cm]
    \bigzero&L_j^\text{ek}&L_j^\text{e}\end{pmatrix}\;.
\end{align}
$L_j^\text{k}$,  $L_j^\text{ke}$ and
$L_j^\text{e}$  can now be chosen freely without changing the SLDCRB. Note that $L_j^\text{se}$ is the zero matrix by definition, see Eq.~\eqref{eq:SLDdef}. We denote the orthonormal basis vectors for the extended space $H_{\text{e}}$
by $\{\ket{e_n}\}$. For example, the
notation $L^\text{sk}$ denotes the matrix with elements
$\bra{s_n}\mathcal{L}\ket{k_m}$. When considering the full SLD operators, we shall not use any superscripts. We use the notation $L$ for SLD operators where any free elements are chosen to be 0 and $\mathcal{L}$ for SLD operators where we choose the free elements ourselves.

We are now in a position to state our main results. To illustrate the essential elements of when this extended space is important, we shall focus on two-parameter estimation.
\begin{enumerate}
\item Result A:  The projection of the commutator of the SLD operators on the support space equal to 0, i.e. $\bra{s_n} [L_1,L_2]\ket{s_m}=0$, is necessary but not sufficient for the existence of an extended SLD operator such that $ [\mathcal{L}_1,\mathcal{L}_2]=0$. Note that the necessary part follows from many existing results, e.g. Ref.~\cite{conlon2022gap}. Our contribution is to present an explicit example where this condition is not sufficient. 

\item Result B: $\bra{s_n}[L_1,L_2]\ket{k_m} = 0$ is neither a necessary nor sufficient condition for the attainability of the SLDCRB. We present an example where the unextended SLD operators do not commute on the support kernel space $\bra{s_n}[L_1,L_2]\ket{k_m} \neq 0$, but the extended SLD operators do commute $[\mathcal{L}_1,\mathcal{L}_2]  = 0$.  We then present an example where $\bra{s_n}[L_1,L_2]\ket{k_m} = 0$, but it is not possible to choose the extended SLD operators such that they commute.

%
%
   We also present an example with three pairs of SLD operators corresponding to the same model $(L_1,L_2)$, $(\mathcal{L}_1,\mathcal{L}_2)$
   and $(\mathcal{\tilde{L}}_1,\mathcal{\tilde{L}}_2)$. For this example we find that $\bra{s_n}[L_1,L_2]\ket{k_m} \neq 0$, $\bra{s_n}[\mathcal{L}_1,\mathcal{L}_2]\ket{k_m} = 0$, \newline\noindent\mbox{$\bra{k_n}[\mathcal{L}_1,\mathcal{L}_2]\ket{k_m} \neq 0$} and  $[\mathcal{\tilde{L}}_1,\mathcal{\tilde{L}}_2]  = 0$. 
   
\item Result C: \mbox{$\bra{k_n} [L_1,L_2]\ket{k_m}=0$} is neither a necessary nor sufficient condition for the attainability of the SLDCRB. Our first example, used for result A above, also demonstrates that the projection of the commutator of the SLD operators on the kernel space being equal to 0,  \mbox{$\bra{k_n} [L_1,L_2]\ket{k_m} = 0$}, is not a sufficient condition for the attainability of the SLDCRB. In this example \mbox{$\bra{k_n} [L_1,L_2]\ket{k_m} = 0$}, but it is not possible to find extended SLD operators which commute. We also present an example where \mbox{$\bra{k_n} [L_1,L_2]\ket{k_m} \neq 0$} but \mbox{$[\mathcal{L}_1,\mathcal{L}_2]=0$}. 
   
  \end{enumerate}

The above results demonstrate the importance of considering the extended space. However, the physical reasoning behind the importance of the extended space may not be immediately clear. A partial explanation comes from the Naimark extension which states that a POVM can always be converted to a projective measurement in a larger Hilbert space~\cite{neumark1943spectral}. Note that we are considering the direct sum Naimark extension ($\tilde{H}=H\oplus H_{\text{e}}$), but the tensor product extension is also possible ($H\otimes H_{\text{e}}$). In the direct sum picture, the simplest Naimark extension is to find the extended space where the SLD operators commute, so that we can measure along the eigenvectors of the SLD operatorm, see e.g Appendix C.2 of Ref.~\cite{conlon2022gap}. The tensor product extension involves a unitary matrix which acts on the quantum state and an ancilla state followed by a projective measurement~\cite{chen2007ancilla}.

It is possible to consider the SLD operators without the extended space, and with the kernel terms set to zero, i.e. $L_{i}$ in our notation. In this case, the saturation conditions involve finding POVMs that must satisfy theorems 1 and 2 of Yang\textit{ et al}~\cite{yang2019optimal}. By considering the extended space, we don't need to concern ourselves with the POVMs. Once, commuting SLD operators in the extended space have been found, the optimal measurement is given by the joint eigenvectors of the extended SLD operators when projected back into the original space.

We know that $\mathcal{C}_\text{S}\leq\mathcal{C}_\text{H}\leq\mathcal{C}_\text{NH}\leq\mathcal{C}_\text{MI}$. In order for $\mathcal{C}_\text{S}=\mathcal{C}_\text{MI}$, we therefore require $\mathcal{C}_\text{S}=\mathcal{C}_\text{H}=\mathcal{C}_\text{NH}$. The condition for equality between $\mathcal{C}_\text{S}$ and $\mathcal{C}_\text{NH}$ as proven in Ref.~\cite{conlon2022gap}, therefore shows that the SLD operators commuting on the support of $\mathcal{S}$ is a necessary condition for $\mathcal{C}_\text{S}=\mathcal{C}_\text{MI}$. Result A shows that this is not a sufficient condition. Results B and C show that the SLD operators not commuting on the kernel-support or kernel-kernel subspace does not imply that the SLD operators cannot be chosen to commute on an extended space. To sum up
\begin{itemize}
\item $\bra{s_n}[L_1,L_2]\ket{s_m} = 0$ is necessary but not sufficient (Example A)
\item $\bra{s_n}[L_1,L_2]\ket{s_m} = 0$  and $\bra{s_n}[L_1,L_2]\ket{k_m} = 0$ is not necessary (Example B/D)
\item $\bra{s_n}[\mathcal{L}_1,\mathcal{L}_2]\ket{s_m} = 0$  and $\bra{s_n}[\mathcal{L}_1,\mathcal{L}_2]\ket{k_m} = 0$ is not sufficient (Example E)
\item $\bra{s_n}[L_1,L_2]\ket{s_m} = 0$  and $\bra{k_n}[L_1,L_2]\ket{k_m} = 0$ is not necessary (Example C/D) and not sufficient (Example A)
\end{itemize}

\subsection{Conditions for Equality}
\label{sec:equality}
Finally, we summarise the conditions for equality between the different bounds: the most informative bound, the NHCRB, the HCRB and the SLDCRB. These conditions are presented in Table.~\ref{T1}. In this table $X_{\text{NH},i}$ and $X_{\text{H},i}$ represent Hermitian matrices which are the solutions to the optimisation problem found in the NHCRB and HCRB respectively. $\text{trAbs\{A\}}$ represents the sum of the absolute values of the matrix $A$. $\{\Pi_{k}^*\}$ represents the optimal POVM as described in section~\ref{secCMI}. $[\mathcal{L}_{i}\,,\,\mathcal{L}_{j}]=0,\,\forall i,j$ is a necessary and sufficient condition for $\mathcal{C}_\text{S}=\mathcal{C}_\text{MI}$. In the remainder of this paper we shall ask under what circumstances we can choose $[\mathcal{L}_{i}\,,\,\mathcal{L}_{j}]=0,\,\forall i,j$.

It should be noted that while, throughout this manuscript we discuss the conditions for $\mathcal{C}_\text{S}=\mathcal{C}_\text{MI}$, it may be more accurate to say that we are examining the conditions for when the classical Fisher information matrix is equal to the quantum Fisher information matrix. For well behaved problems the two are equivalent.

\begin{table}
\begin{tabular}{ | M{3em} |  M{4.2cm}| M{4.8cm}| M{4.2cm} | } 
  \hline
 &  $\mathcal{C}_\text{NH}$& $\mathcal{C}_\text{H}$& $\mathcal{C}_\text{S}$ \\ 
 \hline
 $\mathcal{C}_\text{MI}$ & $$X_{\text{NH},i}=\sum_{k}(\hat{\theta}_{i,k}^*-\theta_{i})\Pi_{k}^*$$~\cite{conlon2021efficient} & $$X_{\text{H},i}=\sum_{k}(\hat{\theta}_{i,k}^*-\theta_{i})\Pi_{k}^*$$ \& $\mathbb{L}_{\text{H},ij}=\mathbb{L}_{\text{H},ji}$, $ \mathbb{L}_{\text{H}}\geq {X_{\text{H}}} X_{\text{H}}^\intercal$ ~\cite{conlon2021efficient}
 &$$[\mathcal{L}_{i}\,,\,\mathcal{L}_{j}]=0,\,\forall i,j$$~\cite{suzuki2020quantum}  \\ 
   \hline
$\mathcal{C}_\text{NH}$   &\diagbox[innerwidth=4.2cm,height=30pt,dir=SW]{}{} & $\mathbb{L}_{\text{H},ij}=\mathbb{L}_{\text{H},ji}$, $ \mathbb{L}_{\text{H}}\geq {X_{\text{H}}} X_{\text{H}}^\intercal$~\cite{conlon2021efficient}&$$\text{trAbs}\{\mathcal{S}[L_{i}\,,\,L_{j}]\}=0,\,\forall i,j$$~\cite{conlon2022gap}  \\ 
   \hline
  $\mathcal{C}_\text{H}$& \diagbox[innerwidth=4.2cm,height=30pt,dir=SW]{}{} &\diagbox[innerwidth=4.8cm,height=30pt,dir=SW]{}{}  & $$\text{tr}\{\mathcal{S}[L_{i}\,,\,L_{j}]\}=0,\,\forall i,j$$~\cite{ragy2016compatibility} \\ 
   \hline
\end{tabular}
\caption{\label{T1}\textbf{Equality conditions for different Cram{\'{e}}r--Rao bounds.}}
\end{table}

\section{Examples supporting our results}
\subsection{Example A}
\label{egA}
Our first example is a rank 3 quantum state in a 4-dimensional Hilbert space. Consider the model with
  \begin{align}
    S^\text{s}=\frac{1}{3}\begin{pmatrix}1&0&0\\0&1&0\\0&0&1\end{pmatrix} \;,
  \end{align}
and the following derivatives 
  \begin{align}
  S_1
  = \frac{1}{3}\begin{pmatrix}\begin{matrix}1&0&0\\0&1&0\\0&0&-2\end{matrix}&\rvline&\begin{matrix}0\\\frac{1}{2}\\0\end{matrix}\\\hline
    \begin{matrix}0&\frac{1}{2}&0\end{matrix}&\rvline&0
   \end{pmatrix}\;,\quad\text{and}\quad
   S_2
  =\frac{1}{3} \begin{pmatrix}\begin{matrix}0&1&0\\1&0&1\\0&1&0\end{matrix}&\rvline&\begin{matrix}0\\0\\-\frac{3}{2}\end{matrix}\\\hline
    \begin{matrix}0&0&-\frac{3}{2}\end{matrix}&\rvline&0
    \end{pmatrix}\;.
\end{align}
This model gives rise to the following SLD operators
  \begin{align}
 \mathcal{L}_1
  = \begin{pmatrix}\begin{matrix}1&0&0\\0&1&0\\0&0&-2\end{matrix}&\rvline&\begin{matrix}0\\1\\0\end{matrix}\\\hline
    \begin{matrix}0&1&0\end{matrix}&\rvline&k_1\end{pmatrix}\;,\quad\text{and}\quad
  \mathcal{L}_2
  = \begin{pmatrix}\begin{matrix}0&1&0\\1&0&1\\0&1&0\end{matrix}&\rvline&\begin{matrix}0\\0\\-3\end{matrix}\\\hline
    \begin{matrix}0&0&-3\end{matrix}&\rvline&k_2\\
    \end{pmatrix}\;,
\end{align}
where $k_1$ and $k_2$ are free parameters. Note that, it is not necessary to present the derivatives and SLD operators, we do so only for completeness. Note also that if $k_1=k_2=0$, $\mathcal{L}_1=L_1$ and $\mathcal{L}_2=L_2$. For this model we find that the quantum Fisher information matrix is given by
\begin{equation}
J_{\text{S}}=\begin{pmatrix}\frac{7}{3}&0\\
0&\frac{13}{3}\end{pmatrix}\;,
\end{equation}
and the associated SLDCRB is given by
\begin{equation}
\mathcal{C}_\text{S}=\frac{60}{91}\;.
\end{equation}
For this example we have $\text{tr}\{S[L_1,L_2]\}=0$, so therefore $\mathcal{C}_\text{S}=\mathcal{C}_\text{H}$. Additionally $\text{trAbs}{S[L_1,L_2]}=0$, so therefore $\mathcal{C}_\text{S}=\mathcal{C}_\text{NH}$~\cite{conlon2022gap}.

For this example, we see that
\begin{align}
[\mathcal{L}_1,\mathcal{L}_2]=\begin{pmatrix}\bigzero&\rvline&\begin{matrix}-1\\k_2\\5+3k_1\end{matrix}\\\hline
    \begin{matrix}\hspace{0.2cm}1&\hspace{0.3cm}-k_2&-5-3k_1\end{matrix}&\rvline&0\end{pmatrix}\;.
\end{align}
Therefore, there is no way to choose $\bra{s_n}[\mathcal{L}_1,\mathcal{L}_2]\ket{k_m}=0$. By choosing extended SLD operators as in Eq.~\eqref{eqSLDopextended}, the elements of the extended space do not influence $\bra{s_n}[\mathcal{L}_1,\mathcal{L}_2]\ket{k_m}$ terms. This example shows that \mbox{$\bra{s_n}[\mathcal{L}_1,\mathcal{L}_2]\ket{s_m}=0$} is not a sufficient condition for the attainability of the SLDCRB. This example also shows that $\bra{s_n}[\mathcal{L}_1,\mathcal{L}_2]\ket{s_m}=0$ and $\bra{k_n}[\mathcal{L}_1,\mathcal{L}_2]\ket{k_m}=0$ is not a sufficient condition for the attainability of the SLDCRB.

This example also offers a simple analytic counter-example demonstrating that the NHCRB is not necessarily attainable for two parameter estimation. This is stated as the following corollary
\begin{corollary}
For estimating two parameters the NHCRB is not always a tight bound, i.e. it is possible to find examples where $\mathcal{C}_\text{MI}>\mathcal{C}_\text{NH}$.
\end{corollary}
\begin{proof}
For the example considered above $\text{trAbs}{S[L_1,L_2]}=0$, therefore $\mathcal{C}_\text{NH}=\mathcal{C}_\text{S}$~\cite{conlon2022gap}. However, for the same example, it is not possible to choose $[\mathcal{L}_1,\mathcal{L}_2]=0$, therefore $\mathcal{C}_\text{MI}>\mathcal{C}_\text{S}$~\cite{suzuki2020quantum}. Therefore, $\mathcal{C}_\text{MI}>\mathcal{C}_\text{NH}$.
\end{proof}

A numerical search for the optimal POVM demonstrates this result. We have $\mathcal{C}_\text{S}=\mathcal{C}_\text{NH}\approx0.6593$, whereas the optimal POVM we could find has a variance of 0.6640.

\subsection{Example  B}
\label{egB}
We next consider an example with $r=2$, $d=3$ which demonstrates the necessity of choosing the free elements of the SLD operators appropriately. We consider the following model
  \begin{align}
    S^\text{s}=\frac{1}{2}\begin{pmatrix}1&0\\0&1\end{pmatrix}   \;,
  \end{align}
and the following derivatives
  \begin{align}
S_1
  = \frac{1}{2}\begin{pmatrix}\begin{matrix}0&1\\[0.1cm]1&0\end{matrix}&\rvline&\begin{matrix}\frac{1}{2}\\[0.1cm]\frac{1}{2}\end{matrix}\\[0.1cm]
    \hline
    \begin{matrix}\frac{1}{2}&\frac{1}{2}\end{matrix}&\rvline&0\end{pmatrix},\quad\text{and}\quad
S_2
  = \frac{\mathrm{1}}{2}\begin{pmatrix}\begin{matrix}1&1\\[0.1cm]1&-1\end{matrix}&\rvline&\begin{matrix}-\frac{1}{2}\\[0.1cm]\frac{1}{2}\end{matrix}\\
    \hline
    \begin{matrix}-\frac{1}{2}&\frac{1}{2}\end{matrix}&\rvline&0\end{pmatrix}\,\;.
\end{align}
This model gives rise to the following SLD operators
  \begin{align}
\mathcal{L}_1
  = \begin{pmatrix}\begin{matrix}0&1\\1&0\end{matrix}&\rvline&\begin{matrix}1\\1\end{matrix}\\
    \hline
    \begin{matrix}1&1\end{matrix}&\rvline&k_1\end{pmatrix},\quad\text{and}\quad
 \mathcal{L}_2
  = \begin{pmatrix}\begin{matrix}1&1\\1&-1\end{matrix}&\rvline&\begin{matrix}-1\\1\end{matrix}\\
    \hline
    \begin{matrix}-1&1\end{matrix}&\rvline&k_2\end{pmatrix}\,,
\end{align}
where, as before, $k_1$ and $k_2$ are free parameters, commonly taken to be 0. For this model we find that the quantum Fisher information matrix is
\begin{equation}
J_{\text{S}}=\begin{pmatrix}2&1\\
1&3\end{pmatrix}\;,
\end{equation}
and the associated SLDCRB is given by
\begin{equation}
\mathcal{C}_\text{S}=1\;.
\end{equation}
Note that for this example we have $\text{tr}\{S[L_1,L_2]\}=0$, so therefore $\mathcal{C}_\text{S}=\mathcal{C}_\text{H}$. Additionally $\text{trAbs}{S[L_1,L_2]}=0$, so therefore $\mathcal{C}_\text{S}=\mathcal{C}_\text{NH}$~\cite{conlon2022gap}.

For this example, we find that
\begin{equation}
[\mathcal{L}_1,\mathcal{L}_2]=\begin{pmatrix}\bigzero&\rvline&\begin{matrix}-1+k_1+k_2\\[0.1cm]-1+k_1+k_2\end{matrix}\\
    \hline
    \begin{matrix}1-k_1-k_2&1+k_1-k_2\end{matrix}&\rvline&0\end{pmatrix}\,\;.
\end{equation}
We see that if we choose $k_1=k_2=0$, then $[L_1,L_2]=[\mathcal{L}_1,\mathcal{L}_2]\neq0$. However, by choosing $k_1=0,k_2=1$, we can make $[\mathcal{L}_1,\mathcal{L}_2]=0$. This example shows that $\bra{s_n}[L_1,L_2]\ket{s_m}=0$ and $\bra{s_n}[L_1,L_2]\ket{k_m}=0$ is not a necessary condition for the attainability of the SLDCRB. For this example, we can construct a 3-outcome optimal POVM that saturates the SLDCRB.

\subsection{Example  C}
\label{egC}
Consider the $r=2$, $d=4$ model with
  \begin{align}
  \label{rhoExC}
    S^\text{s}=\frac{1}{2}\begin{pmatrix}1&0\\0&1\end{pmatrix}   \;,
  \end{align}
and the following derivatives 
  \begin{align}
S_1
  = \frac{1}{2}\begin{pmatrix}\begin{matrix}1&0\\0&-1\end{matrix}&\rvline&\begin{matrix}\frac{1}{2}&0\\0&\frac{1}{2}\end{matrix}\\
    \hline
    \begin{matrix}\frac{1}{2}&0\\0&\frac{1}{2}\end{matrix}&\rvline&\bigzero\end{pmatrix},\quad\text{and}\quad
S_2
  = \frac{\mathrm{i}}{2}\begin{pmatrix}\begin{matrix}0&1\\-1&0\end{matrix}&\rvline&\begin{matrix}0&\frac{1}{2}\\\frac{1}{2}&0\end{matrix}\\
    \hline
    \begin{matrix}0&-\frac{1}{2}\\-\frac{1}{2}&0\end{matrix}&\rvline&\bigzero\end{pmatrix}\,\;.
\end{align}
This model gives rise to the following SLD operators
  \begin{align}
L_1
  = \begin{pmatrix}\begin{matrix}1&0\\0&-1\end{matrix}&\rvline&\begin{matrix}1&0\\0&1\end{matrix}\\
    \hline
    \begin{matrix}1&0\\0&1\end{matrix}&\rvline&\bigzero\end{pmatrix},\quad\text{and}\quad
 L_2
  = \begin{pmatrix}\begin{matrix}0&\mathrm{i}\\-\mathrm{i}&0\end{matrix}&\rvline&\begin{matrix}0&\mathrm{i}\\\mathrm{i}&0\end{matrix}\\
    \hline
    \begin{matrix}0&-\mathrm{i}\\-\mathrm{i}&0\end{matrix}&\rvline&\bigzero\end{pmatrix}\,.
\end{align}
For this model we find that the quantum Fisher information matrix is given by
\begin{equation}
J_{\text{S}}=\begin{pmatrix}2&0\\
0&2\end{pmatrix}\;,
\end{equation}
and the associated SLDCRB is given by
\begin{equation}
\label{eqSLDCRBExC}
\mathcal{C}_\text{S}=1\;.
\end{equation}
Once again, for this example we have $\text{tr}\{S[L_1,L_2]\}=0$ and $\text{trAbs}{S[L_1,L_2]}=0$, so therefore $\mathcal{C}_\text{S}=\mathcal{C}_\text{H}=\mathcal{C}_\text{NH}$. 

The two operators do not commute:
  \begin{equation}  [L_1,L_2]   =
  \begin{pmatrix}\bigzero&\rvline&\bigzero\\
    \hline
\bigzero&\rvline&
  \begin{matrix}0&2\mathrm{i}\\2\mathrm{i}&0\end{matrix}\end{pmatrix}\,.
\end{equation}

 For the same model, consider the following choice of SLD operators
  \begin{align}
  \label{SLDopCommuteExC}
  \mathcal{L}_1
  = \begin{pmatrix}\begin{matrix}1&0\\0&-1\end{matrix}&\rvline&\begin{matrix}1&0\\0&1\end{matrix}\\
    \hline
    \begin{matrix}1&0\\0&1\end{matrix}&\rvline&
  \begin{matrix}1&0\\0&-1\end{matrix}\end{pmatrix}\;,\quad\text{and}\quad
  \mathcal{L}_2
  = \begin{pmatrix}\begin{matrix}0&\mathrm{i}\\-\mathrm{i}&0\end{matrix}&\rvline&\begin{matrix}0&\mathrm{i}\\\mathrm{i}&0\end{matrix}\\
    \hline
    \begin{matrix}0&-\mathrm{i}\\-\mathrm{i}&0\end{matrix}&\rvline&
  \begin{matrix}0&-\mathrm{i}\\\mathrm{i}&0\end{matrix}\end{pmatrix}\,.
\end{align}
These two operators commute $  [\mathcal{L}_1,\mathcal{L}_2]=0 $.

 Additionally, it is worth noting that there are many possible choices for $  [\mathcal{L}_1,\mathcal{L}_2]=0 $. Another example is the following
  \begin{align}
  \mathcal{\tilde{L}}_1
  = \begin{pmatrix}\begin{matrix}1&0\\0&-1\end{matrix}&\rvline&\begin{matrix}1&0\\0&1\end{matrix}&\rvline&\bigzero\\
    \hline
    \begin{matrix}1&0\\0&1\end{matrix}&\rvline&
    \begin{matrix}\bigzero\end{matrix}&\rvline&\begin{matrix}\sqrt{2}&0\\0&0\end{matrix}\\
    \hline
  \bigzero&\rvline&\begin{matrix}\sqrt{2}&0\\0&0\end{matrix}&\rvline&\bigzero\end{pmatrix}\;,\quad\text{and}\quad  \mathcal{\tilde{L}}_2
  = \begin{pmatrix}\begin{matrix}0&\mathrm{i}\\-\mathrm{i}&0\end{matrix}&\rvline&\begin{matrix}0&\mathrm{i}\\\mathrm{i}&0\end{matrix}&\rvline&\bigzero\\
    \hline
    \begin{matrix}0&-\mathrm{i}\\-\mathrm{i}&0\end{matrix}&\rvline&
    \begin{matrix}\bigzero\end{matrix}&\rvline&\begin{matrix}0&0\\\mathrm{i}\sqrt{2}&0\end{matrix}\\
    \hline
  \bigzero&\rvline&\begin{matrix}0&-\mathrm{i}\sqrt{2}\\0&0\end{matrix}&\rvline&\bigzero\end{pmatrix}\,.
\end{align}
Again, these two operators commute $  [\mathcal{\tilde{L}}_1,\mathcal{\tilde{L}}_2]=0 $. It is easy to verify that all three of these SLD operators give an SLDCRB of 1. This example shows that $\bra{s_n}[L_1,L_2]\ket{s_m}=0$ and $\bra{k_n}[L_1,L_2]\ket{k_m}=0$ is not a necessary condition for the attainability of the SLDCRB. We present a physical model and a measurement saturating the SLDCRB for this example in Appendices~\ref{apen:physicalstate} and \ref{apenMeasC} respectively.

\subsection{Example D}
\label{egD}
 Consider the $r=2$, $d=6$ model with
  \begin{align}
    S^\text{s}=\frac{1}{2}\begin{pmatrix}1&0\\0&1\end{pmatrix}    
  \end{align}
and derivatives equal to
  \begin{align}
S_1
  = \frac{1}{2}\begin{pmatrix}\begin{matrix}1&0\\0&-1\end{matrix}&\rvline&\begin{matrix}\frac{1}{2}&0&0&0\\0&0&0&\frac{1}{2}\end{matrix}\\
    \hline
    \begin{matrix}\frac{1}{2}&0\\0&0\\0&0\\0&\frac{1}{2}\end{matrix}&\rvline&\bigzero\end{pmatrix}\;,\quad\text{and}\quad
S_2
  = \frac{1}{2}\begin{pmatrix}\begin{matrix}0&1\\1&0\end{matrix}&\rvline&\begin{matrix}0&0&0&\frac{1}{2}\\[0.1cm]-\frac{1}{2}&-\frac{1}{2}&\frac{1}{2}&\frac{1}{2}\end{matrix}\\[0.1cm]
    \hline
    \begin{matrix}0&-\frac{1}{2}\\[0.1cm]0&-\frac{1}{2}\\[0.1cm]0&\frac{1}{2}\\[0.1cm]\frac{1}{2}&\frac{1}{2}\end{matrix}&\rvline&\bigzero\end{pmatrix}\,.
\end{align}
This problem gives rise to the following two SLD operators 
  \begin{align}
L_1
  = \begin{pmatrix}\begin{matrix}1&0\\0&-1\end{matrix}&\rvline&\begin{matrix}1&0&0&0\\0&0&0&1\end{matrix}\\
    \hline
    \begin{matrix}1&0\\0&0\\0&0\\0&1\end{matrix}&\rvline&\bigzero\end{pmatrix}\;,\quad\text{and}\quad
L_2
  = \begin{pmatrix}\begin{matrix}0&1\\1&0\end{matrix}&\rvline&\begin{matrix}0&0&0&1\\-1&-1&1&1\end{matrix}\\
    \hline
    \begin{matrix}0&-1\\0&-1\\0&1\\1&1\end{matrix}&\rvline&\bigzero\end{pmatrix}\,.
\end{align}
For this model we find that the quantum Fisher information matrix is given by
\begin{equation}
J_{\text{S}}=\begin{pmatrix}2&\frac{1}{2}\\[0.1cm]
\frac{1}{2}&\frac{7}{2}\end{pmatrix}\;,
\end{equation}
and the associated SLDCRB is given by
\begin{equation}
\label{eqSLDexD}
\mathcal{C}_\text{S}=\frac{22}{27}\;.
\end{equation}
As before, we have $\text{tr}\{S[L_1,L_2]\}=0$ and $\text{trAbs}{S[L_1,L_2]}=0$, so therefore $\mathcal{C}_\text{S}=\mathcal{C}_\text{H}=\mathcal{C}_\text{NH}$. The two SLD operators do not commute on either the support-kernel or kernel subspaces:
 \begin{align} [L_1,L_2]   =
  \begin{pmatrix}\bigzero&\rvline&\begin{matrix}0&0&0&0\\0&1&-1&-1\end{matrix}\\
    \hline
    \begin{matrix}0&0\\0&-1\\0&1\\0&1\end{matrix}&\rvline&
  \begin{matrix}0&0&0&2\\0&0&0&1\\0&0&0&-1\\-2&-1&1&0\end{matrix}\end{pmatrix}\,.
\end{align}

 We now show that we can choose the kernel of the SLD operators such that they commute on the support-kernel subspace but not on the kernel subspace. We are free to choose the SLD operators such that
  \begin{align}
  \mathcal{L}_1
  = \begin{pmatrix}\begin{matrix}1&0\\0&-1\end{matrix}&\rvline&\begin{matrix}1&0&0&0\\0&0&0&1\end{matrix}\\
    \hline
    \begin{matrix}1&0\\0&0\\0&0\\0&1\end{matrix}&\rvline&\begin{matrix}1&0&0&0\\0&-1&0&0\\0&0&-1&0\\0&0&0&-1\end{matrix}\end{pmatrix}\;,\quad\text{and}\quad
  \mathcal{L}_2
  = \begin{pmatrix}\begin{matrix}0&1\\1&0\end{matrix}&\rvline&\begin{matrix}0&0&0&1\\-1&-1&1&1\end{matrix}\\
    \hline
    \begin{matrix}0&-1\\0&-1\\0&1\\1&1\end{matrix}&\rvline&\begin{matrix}0&0&0&-1\\0&0&0&0\\0&0&0&0\\-1&0&0&0\end{matrix}\end{pmatrix}\,.
\end{align}
The two operators now commute on the support-kernel subspace but not the kernel subspace:
\begin{align}  [\mathcal{L}_1,\mathcal{L}_2]   =
  \begin{pmatrix}\bigzero&\rvline&\bigzero\\
    \hline
   \bigzero&\rvline&
  \begin{matrix}0&0&0&0\\0&0&0&1\\0&0&0&-1\\0&-1&1&0\end{matrix}\end{pmatrix}\,.
\end{align}

 Finally, we show that there exists an extension of the SLD operators such that they commute on the extended space. We can choose the SLD operators as   
\begin{align}
\label{EqExDcommute}
  \mathcal{\tilde{L}}_1
  = \begin{pmatrix}\begin{matrix}1&0\\0&-1\end{matrix}&\rvline&\begin{matrix}1&0&0&0\\0&0&0&1\end{matrix}&\rvline&\bigzero\\
    \hline
    \begin{matrix}1&0\\0&0\\0&0\\0&1\end{matrix}&\rvline&\begin{matrix}1&0&0&0\\0&-1&0&0\\0&0&-1&0\\0&0&0&-1\end{matrix}&\rvline&\begin{matrix}0&0\\0&-1\\1&0\\0&0\end{matrix}\\
     \hline
    \bigzero&\rvline&\begin{matrix}0&0&1&0\\0&-1&0&0\end{matrix}&\rvline&\begin{matrix}0&-1\\-1&0\end{matrix}
     \end{pmatrix}\;,\quad\text{and}\quad
  \mathcal{\tilde{L}}_2
  = \begin{pmatrix}\begin{matrix}0&1\\1&0\end{matrix}&\rvline&\begin{matrix}0&0&0&1\\-1&-1&1&1\end{matrix}&\rvline&\bigzero\\
    \hline
    \begin{matrix}0&-1\\0&-1\\0&1\\1&1\end{matrix}&\rvline&\begin{matrix}0&0&0&-1\\0&0&0&0\\0&0&0&0\\-1&0&0&0\end{matrix}&\rvline&\begin{matrix}0&0\\0&0\\0&0\\1&1\end{matrix}\\
    \hline
    \bigzero&\rvline&\begin{matrix}0&0&0&1\\0&0&0&1\end{matrix}&\rvline&\bigzero
\end{pmatrix}\,.
\end{align}

Now we find $[ \mathcal{\tilde{L}}_1, \mathcal{\tilde{L}}_2]=0$. This example shows that neither $\bra{s_n}[L_1,L_2]\ket{s_m}=0$ and $\bra{k_n}[L_1,L_2]\ket{s_m}=0$, nor $\bra{s_n}[L_1,L_2]\ket{s_m}=0$ and $\bra{k_n}[L_1,L_2]\ket{k_m}=0$, is a necessary condition for the attainability of the SLDCRB. In Appendix~\ref{apenMeasD}, we present a measurement based on the commuting SLD operators which saturates the SLDCRB.

\subsection{Example E}
\label{exE}
We now present our final example which has $r=3$, $d=5$.  Consider the model with
  \begin{align}
    S^\text{s}=\frac{1}{3}\begin{pmatrix}1&0&0\\0&1&0\\0&0&1\end{pmatrix}    
  \end{align}
and derivatives equal to
  \begin{align}
S_1
  = \frac{1}{3}\begin{pmatrix}\begin{matrix}1&0&0\\[0.1cm]0&0&0\\[0.1cm]0&0&-1\end{matrix}&\rvline&\begin{matrix}\frac{1}{2}&0\\[0.1cm]0&0\\[0.1cm]0&\frac{1}{2}\end{matrix}\\[0.6cm]
    \hline
    \begin{matrix}\frac{1}{2}&0&0\\[0.1cm]0&0&\frac{1}{2}\end{matrix}&\rvline&\bigzero\end{pmatrix}\;,\quad\text{and}\quad
S_2
  = \frac{1}{3}\begin{pmatrix}\begin{matrix}0&1&0\\[0.1cm]1&0&1\\[0.1cm]0&1&0\end{matrix}&\rvline&\begin{matrix}\frac{1}{2}&\frac{1}{2}\\[0.1cm]-\frac{1}{2}&\frac{1}{2}\\[0.1cm]\frac{1}{2}&-\frac{1}{2}\end{matrix}\\[0.6cm]
    \hline
    \begin{matrix}\frac{1}{2}&-\frac{1}{2}&\frac{1}{2}\\[0.1cm]\frac{1}{2}&\frac{1}{2}&-\frac{1}{2}\end{matrix}&\rvline&\bigzero\end{pmatrix}\,.
\end{align}
This problem gives rise to the following two SLD operators 
  \begin{align}
L_1
  = \begin{pmatrix}\begin{matrix}1&0&0\\0&0&0\\0&0&-1\end{matrix}&\rvline&\begin{matrix}1&0\\0&0\\0&1\end{matrix}\\
    \hline
    \begin{matrix}1&0&0\\0&0&1\end{matrix}&\rvline&\bigzero\end{pmatrix}\;,\quad\text{and}\quad
L_2
  = \begin{pmatrix}\begin{matrix}0&1&0\\1&0&1\\0&1&0\end{matrix}&\rvline&\begin{matrix}1&1\\-1&1\\1&-1\end{matrix}\\
    \hline
    \begin{matrix}1&-1&1\\1&1&-1\end{matrix}&\rvline&\bigzero\end{pmatrix}\,.
\end{align}
For this model the quantum Fisher information matrix is equal to
\begin{equation}
J_{\text{S}}=\begin{pmatrix}\frac{4}{3}&0\\
0&\frac{10}{3}\end{pmatrix}\;,
\end{equation}
and the associated SLDCRB is given by
\begin{equation}
\mathcal{C}_\text{S}=\frac{21}{20}\;.
\end{equation}
As before, we have $\text{tr}\{S[L_1,L_2]\}=0$ and $\text{trAbs}{S[L_1,L_2]}=0$, so therefore $\mathcal{C}_\text{S}=\mathcal{C}_\text{H}=\mathcal{C}_\text{NH}$. The two SLD operators do not commute on the support-kernel subspace:
 \begin{align} [L_1,L_2]   =
  \begin{pmatrix}\bigzero&\rvline&\begin{matrix}1&1\\-1&-1\\-1&1\end{matrix}\\
    \hline
    \begin{matrix}-1&1&1\\-1&1&-1\end{matrix}&\rvline&
 \bigzero\end{pmatrix}\,.
\end{align}

We can now consider the following extended SLD operators
\begin{align}
\mathcal{L}_1
  = \begin{pmatrix}\begin{matrix}1&0&0\\0&0&0\\0&0&-1\end{matrix}&\rvline&\begin{matrix}1&0\\0&0\\0&1\end{matrix}&\rvline&\bigzero\\
    \hline
    \begin{matrix}1&0&0\\0&0&1\end{matrix}&\rvline&L_1^\text{k}&\rvline&L_1^\text{ke}    \\
    \hline
\bigzero&\rvline&L_1^\text{ek}&\rvline&L_1^\text{e}
    \end{pmatrix}\;,\quad\text{and}\quad
\mathcal{L}_2
  = \begin{pmatrix}\begin{matrix}0&1&0\\1&0&1\\0&1&0\end{matrix}&\rvline&\begin{matrix}1&1\\-1&1\\1&-1\end{matrix}&\rvline&\bigzero\\
    \hline
    \begin{matrix}1&-1&1\\1&1&-1\end{matrix}&\rvline&L_2^\text{k}&\rvline&L_2^\text{ke}\\    \hline
\bigzero&\rvline&L_2^\text{ek}&\rvline&L_2^\text{e}\end{pmatrix}\,.
\end{align}
Only terms in the support-kernel and kernel subspaces influence the  $\bra{k_n}[\mathcal{L}_1,\mathcal{L}_2]\ket{s_m}$ terms. In order to make $\bra{k_n}[\mathcal{L}_1,\mathcal{L}_2]\ket{s_m}=0$, we are required to choose the free elements in the kernel space as
\begin{align}
\mathcal{L}_1
  = \begin{pmatrix}\begin{matrix}1&0&0\\0&0&0\\0&0&-1\end{matrix}&\rvline&\begin{matrix}1&0\\0&0\\0&1\end{matrix}&\rvline&\bigzero\\
    \hline
    \begin{matrix}1&0&0\\0&0&1\end{matrix}&\rvline&\begin{matrix}3&2\\2&1\end{matrix}&\rvline&L_1^\text{ke}    \\
    \hline
\bigzero&\rvline&L_1^\text{ek}&\rvline&L_1^\text{e}
    \end{pmatrix}\;,\quad\text{and}\quad
\mathcal{L}_2
  = \begin{pmatrix}\begin{matrix}0&1&0\\1&0&1\\0&1&0\end{matrix}&\rvline&\begin{matrix}1&1\\-1&1\\1&-1\end{matrix}&\rvline&\bigzero\\
    \hline
    \begin{matrix}1&-1&1\\1&1&-1\end{matrix}&\rvline&\begin{matrix}4&2\\2&0\end{matrix}&\rvline&L_2^\text{ke}\\[0.1cm]  \hline
\bigzero&\rvline&L_2^\text{ek}&\rvline&L_2^\text{e}\end{pmatrix}\,.
\end{align}

We next wish to make the two operators commute on the support--extended kernel subspace. Denoting the dimensions of the extended space as $e$, this condition requires
\begin{equation}
L_1^\text{ek}
  = \frac{1}{2}\begin{pmatrix}a_1&a_2&a_3\hdots&a_e\\   a_1&a_2&a_3\hdots&a_e\end{pmatrix}\;.
\end{equation}
and
\begin{equation}
L_2^\text{ek}
  = \begin{pmatrix}a_1&a_2&a_3\hdots&a_e\\    0&0&0\hdots&0\end{pmatrix}\;.
\end{equation}
However, when we impose this condition we find that it is impossible to choose $\bra{k_n}[\mathcal{L}_1,\mathcal{L}_2]\ket{k_m}=0$. Therefore, $\bra{s_n}[\mathcal{L}_1,\mathcal{L}_2]\ket{s_m}=0$ and $\bra{k_n}[\mathcal{L}_1,\mathcal{L}_2]\ket{s_m}=0$ is not a sufficient condition to saturate the SLDCRB.

\section{One-dimensional kernel space}
\label{sec:1D}
In examples A and B, we have $d=3$ and $r=2$, i.e. the kernel is one-dimensional. In one of these examples we can saturate the SLDCRB but not in the other. Hence, one might wonder whether it is possible to provide necessary and sufficient conditions for the attainability of the SLDCRB in this and similar settings. In this section we provide a simple answer to this question, presenting the necessary and sufficient conditions for the attainability of the SLDCRB whenever $d=r+1$, i.e. for one-dimensional kernel spaces. This is stated in theorem~\ref{theoremkernel1}. For this section, let the state $\mathcal{S}$ and SLD operators $L_j$ be
\begin{align}
\label{eqsldJun}
\mathcal{S}=\begin{pmatrix}
S^\text{s}&0\\
0&0
\end{pmatrix},
\quad 
L_j=\begin{pmatrix}
L^\text{s}_j&\ket{\ell_j}\\
\bra{\ell_j}&0
\end{pmatrix}\ (j=1,2)\;.
\end{align}
Here we have written $L^{\text{sk}}_j$ as $\ket{\ell_j}$ to emphasise that these quantities are vectors and to simplify the following. In the following we assume $L_j^\text{s}$ and $\ket{\ell_j}$ are non zero. As before we shall use $\mathcal{L}_j$ to denote the SLD operator when we choose non-zero values for the free elements. We can now state the main result of this section
\begin{theorem}
\label{theoremkernel1}
For a probe state $\mathcal{S}$ of rank $r$ in a $d$ dimensional Hilbert space $\cH$, such that $d=r+1$, the necessary and sufficient condition to saturate the two-parameter SLDCRB is 
that we can choose $\mathcal{L}_j$ such that $[\mathcal{L}_1,\mathcal{L}_2]$ vanishes on the support space and support-kernel space. 
Furthermore, this condition is equivalent to the following three conditions.
\begin{align}
&[L_1^\text{s},L_2^\text{s}]\mbox{ is rank 2}\\
&\Im\langle\ell_1|\ell_2\rangle=0\\
&L_1^\text{s}\ket{\ell_2}-L_2^\text{s}\ket{\ell_1}\in\mathrm{span}\{\ket{\ell_1},\ket{\ell_2}\}\;.
\end{align}
\end{theorem}
\begin{proof}

The commutation relation, for the SLD operators as written in Eq.~\eqref{eqsldJun}, is given by
\begin{equation}
[L_1,L_2]=\begin{pmatrix}
[L_1^\text{s},L_2^\text{s}]+\ket{\ell_1}\bra{\ell_2}-\ket{\ell_2}\bra{\ell_1}\ &L_1^\text{s}\ket{\ell_2}-L_2^\text{s}\ket{\ell_1}\\
\bra{\ell_1}L_2^\text{s}-\bra{\ell_2}L_1^\text{s}&\langle\ell_1|\ell_2\rangle-\langle\ell_2|\ell_1\rangle
\end{pmatrix}\;.
\end{equation}
The condition for commutativity on the support is 
\begin{equation}\label{cond1}
\bra{s_n}[L_1,L_2]\ket{s_m}=0\ \iff\ [L_1^\text{s},L_2^\text{s}]+\ket{\ell_1}\bra{\ell_2}-\ket{\ell_2}\bra{\ell_1}=0\;.
\end{equation}
Since the two $L_j^\text{s}$ are Hermitian, $\mathrm{i}[L^\text{s}_1,L_2^\text{s}]$ is also Hermitian. Furthermore, this commutator is traceless. 
On the other hand, $\mathrm{i}(\ket{\ell_1}\bra{\ell_2}-\ket{\ell_2}\bra{\ell_1})$ is a Hermitian matrix of rank two. 
However, this operator is not necessarily traceless, \newline\noindent\mbox{$\text{tr}\{\ket{\ell_1}\bra{\ell_2}-\ket{\ell_2}\bra{\ell_1}\}=2\mathrm{i}\Im\langle\ell_2|\ell_1\rangle$}. We thus have the following equivalence. 
\begin{equation}\label{cond1_2}
\bra{s_n}[L_1,L_2]\ket{s_m}=0\ \iff\ [L_1^\text{s},L_2^\text{s}]\mbox{ is rank 2 and }\Im\langle\ell_1|\ell_2\rangle=0\;.
\end{equation}
Note that if $[L_1^\text{s},L_2^\text{s}]$ is rank 2, we can always find $\ket{\ell_j}$ satisfying the above condition. Secondly, let us note that if $\bra{s_n}[L_1,L_2]\ket{s_m}=0$, then $\bra{k_n}[L_1,L_2]\ket{k_m}=0$ also holds. 

Next, let us consider an extension with only kernel components. 
\[
\mathcal{L}_j=\begin{pmatrix}
L^\text{s}_j&\ket{\ell_j}\\
\bra{\ell_j}&k_j
\end{pmatrix}\ (j=1,2)\;.
\]
Under the assumption of $\bra{s_n}[L_1,L_2]\ket{s_m}=0$, we immediately see that 
\begin{equation}
[\mathcal{L}_1,\mathcal{L}_2]=0\ \iff\ (L_1^\text{s}-k_1)\ket{\ell_2}-(L_2^\text{s}-k_2)\ket{\ell_1}=0\;.
\end{equation}
$k_1$ and $k_2$ can only be chosen to satisfy the above condition when
\begin{equation}
L_1^\text{s}\ket{\ell_2}-L_2^\text{s}\ket{\ell_1}\in\mathrm{span}\{\ket{\ell_1},\ket{\ell_2}\}\;.
\end{equation}
This cannot be true in general (see Example A). 
However, when dim$\cH$=3, as long as the condition \eqref{cond1} is satisfied, we can always find $k_j$ such that this condition is satisfied (see Example B). 
This is because $\ket{\ell_1},\ket{\ell_2}$ span the whole two dimensional 
space and $L_1^\text{s}\ket{\ell_2}-L_2^\text{s}\ket{\ell_1}$ can be uniquely expanded by them. Note that we have excluded the case where one of the $\ket{\ell_j}$ is the zero vector.\\

Finally, we consider a full extension of the SLDs as 
\begin{equation}
\tilde{\mathcal{L}}_j=\begin{pmatrix}
L^\text{s}_j&\ket{\ell_j}&0\\
\bra{\ell_j}&k_j&\bra{\tilde{k}_j}\\
0&\ket{\tilde{k}_j}&L_j^\text{e}
\end{pmatrix}\ (j=1,2)\;,
\end{equation}
where we use the notation that $L_j^\text{ke}=\bra{\tilde{k}_j}$ to highlight that this quantity is a vector. We use a tilde for this vector to distinguish it from the orthonormal basis vectors which span the kernel space, defined earlier. Under the assumption of $\bra{s_n}[L_1,L_2]\ket{s_m}=0$, 
the commutation relation is 
\begin{equation}
[\tilde{\mathcal{L}}_1,\tilde{\mathcal{L}}_2]=\begin{pmatrix}
0 &(L_1^\text{s}-k_1)\ket{\ell_2}-(L_2^\text{s}-k_2)\ket{\ell_1}&\ket{\ell_1}\bra{\tilde{k}_2}-\ket{\ell_2}\bra{\tilde{k}_1}\\
\bra{\ell_1}(L_2^\text{s}-k_2)-\bra{\ell_2}(L_1^\text{s}-k_1)&0&\bra{\tilde{k}_1}(L_2^\text{e}-k_2)-\bra{\tilde{k}_2}(L_1^\text{e}-k_1) \\
\ket{\tilde{k}_1}\bra{\ell_2}-\ket{\tilde{k}_2}\bra{\ell_1}&(L_1^\text{e}-k_1)\ket{\tilde{k}_2}-(L_2^\text{e}-k_2)\ket{\tilde{k}_1}
&\ket{\tilde{k}_1}\bra{\tilde{k}_2}-\ket{\tilde{k}_2}\bra{\tilde{k}_1}+[L_1^\text{e},L_2^\text{e}]
\end{pmatrix}\;.
\end{equation}
Provided the $\ket{\ell_j}$ are linearly independent, from the support-extension component, we must have $\ket{\tilde{k}_j}=0$. 
Under this condition, we have 
\begin{equation}
[\tilde{\mathcal{L}}_1,\tilde{\mathcal{L}}_2]=0\ \iff\ (L_1^\text{s}-k_1)\ket{\ell_2}-(L_2^\text{s}-k_2)\ket{\ell_1}=0\mbox{ and }[L_1^\text{e},L_2^\text{e}]=0\;.
\end{equation}
This concludes that any extension with extended space does not help to have commutative SLD operators. 
\end{proof}
Finally note that for rank$\mathcal{S}=2$, the last condition in theorem~\ref{theoremkernel1} is automatically satisfied.

\section{Discussion and conclusion}
\label{sec:conc}
This work highlights the role and importance of the kernel and extended spaces in the attainability of the SLDCRB. Ignoring the extended space and requiring the SLD operators to be commuting essentially limits the space of possible measurements to be projective measurements rather than POVMs. Note that there are other, less restrictive conditions that allow for more general POVMS without considering the extended space~\cite{yang2019optimal}. Hence, it is not necessary to consider the extended space, however as our examples demonstrate the extended space can be a very useful tool.



From a practical viewpoint, implementing collective measurements on even a finite number of copies of the probe state is difficult~\cite{roccia2017entangling,parniak2018beating,hou2018deterministic,wu2019experimentally,yuan2020direct,conlon2023approaching,conlon2023discriminating,zhou2023experimental,tian2024minimum}. Therefore, we can expect that the conditions we have presented will be of relevance to practical quantum metrology experiments. We have presented explicit examples which demonstrate the necessity and sufficiency of certain conditions on the SLD operators for the attainability of the SLDCRB. However, many important open questions remain. The main open question is to provide general conditions on when the SLD operators can be chosen to commute, i.e. what distinguishes examples A and E from examples B, C and D? We have provided the necessary and sufficient conditions for two parameter estimation for the case when the kernel space is one-dimensional. We have also presented a simple analytic example showing that the NHCRB is not always a tight bound. This complements the recent numerical results of Ref.~\cite{hayashi2023tight}. 


Finally, it might be tempting to argue that our examples are rather abstract and do not correspond to a physical model. To that end, we provide a physical model for our example C in Appendix~\ref{apen:physicalstate}. 
  
  \section*{Acknowledgements}
This research was funded by the Australian Research Council Centre of Excellence CE170100012, Laureate Fellowship FL150100019 and the Australian Government Research Training Program Scholarship. 

\noindent This research is supported by A*STAR C230917010, Emerging Technology and A*STAR C230917004, Quantum Sensing. 

\noindent JS is partially supported by JSPS KAKENHI Grant Numbers JP21K04919, JP21K11749.

\section*{Code Availability}
The code that support the findings of this study is available from the corresponding author upon reasonable request.
\section*{Competing Interests}
The authors declare no competing interests

  \appendix
\section{Physical model for example C}
\label{apen:physicalstate}
We now present a physical model for example C, using the quantum exponential family. Given two commuting SLD operators, the following provides a general description of the quantum state
\begin{equation}
\label{eqExponentialFam}
\mathcal{S}(\theta_1,\theta_2)=\text{e}^{(\mathcal{L}_1\theta_1+\mathcal{L}_2\theta_2)/2-\psi(\theta)/2}\mathcal{S}_0\text{e}^{(\mathcal{L}_1\theta_1+\mathcal{L}_2\theta_2)/2-\psi(\theta)/2}
\end{equation} 
where we have made the dependence on $\theta_1,\theta_2$ explicit, $\mathcal{S}_0$ is defined by Eq.~\eqref{rhoExC} for our particular problem and $\psi(\theta)=\text{log}(\text{tr}\{\mathcal{S}_0\text{exp}(\mathcal{L}_1\theta_1+\mathcal{L}_2\theta_2)\})$ is a normalisation factor. Note that the extension to multiple parameters follows in a simple manner~\cite{suzuki2020quantum}. Evaluating Eq.~\eqref{eqExponentialFam} and its derivatives at $\theta_1=\theta_2=0$ gives example C as claimed.


\section{Measurement saturating the SLDCRB for example C}
\label{apenMeasC}
In this appendix we present an explicit POVM which saturates the SLDCRB for example C. We use the SLD operators in Eq.~\eqref{SLDopCommuteExC} for this. We shall use the following normalised eigenvectors to construct our POVM
\begin{align}
\ket{\psi_1}=\frac{1}{\sqrt{2}}\begin{pmatrix}0\\-1\\0\\1\end{pmatrix}\;,
\ket{\psi_2}=\frac{1}{\sqrt{2}}\begin{pmatrix}1\\0\\1\\0\end{pmatrix}\;,
\ket{\psi_1}=\frac{1}{2}\begin{pmatrix}-\mathrm{i}\\1\\\mathrm{i}\\1\end{pmatrix}\;,
\ket{\psi_1}=\frac{1}{2}\begin{pmatrix}\mathrm{i}\\1\\-\mathrm{i}\\1\end{pmatrix}\;.
\end{align}
We then use the POVM $\{\Pi_i=\ket{\psi_i}\bra{\psi_i}\}$. We can then calculate the classical Fisher information using this POVM to find that it saturates the SLDCRB, Eq.~\eqref{eqSLDCRBExC} as claimed.

%
%

\section{Measurement saturating the SLDCRB for example D}
\label{apenMeasD}
From the commuting SLD operators in the extended space, Eq.~\eqref{EqExDcommute}, it is possible to obtain a simple projective measurement saturating the SLDCRB. Note that as the matrices $\mathcal{\tilde{L}}_1$ and $\mathcal{\tilde{L}}_2$ are degenerate, their eigenvectors are not unique. Therefore the standard measurement along the eigenvectors of one of the SLD operators will not work in this case. (See e.g. Appendix C.2. of Ref.~\cite{conlon2022gap}.) We have to measure along a set of eigenvectors that is common to both. In this case, we can find the eigenvectors of $\mathcal{\tilde{L}}_1+\mathcal{\tilde{L}}_2$ and use this to construct our POVM. The normalised eigenvectors are
\begin{align}
\ket{\psi_1}=&\frac{1}{2\sqrt{3}}\begin{pmatrix}0\\2\\0\\1\\-1\\-2\\1\\1\end{pmatrix}\;,
\ket{\psi_2}=\frac{1}{\sqrt{30}}\begin{pmatrix}2\\3\\-2\\-1\\1\\3\\1\\1\end{pmatrix}\;,
\ket{\psi_3}=\frac{1}{2\sqrt{5}}\begin{pmatrix}2\\-2\\-2\\-1\\1\\-2\\1\\1\end{pmatrix}\;,
\ket{\psi_4}=\frac{1}{\sqrt{2}}\begin{pmatrix}1\\0\\1\\0\\0\\0\\0\\0\end{pmatrix}\;, \\
\ket{\psi_5}=\frac{1}{2\sqrt{2+\sqrt{2}}}&\begin{pmatrix}0\\0\\0\\1+\sqrt{2}\\1+\sqrt{2}\\0\\-1\\1\end{pmatrix}\;,
\ket{\psi_6}=\frac{1}{2\sqrt{2+\sqrt{2}}}\begin{pmatrix}0\\0\\0\\-1\\-1\\0\\-1-\sqrt{2}\\1+\sqrt{2}\end{pmatrix}\;,
\ket{\psi_7}=\frac{1}{\sqrt{6}}\begin{pmatrix}0\\-1\\0\\1\\-1\\1\\1\\1\end{pmatrix}\;, 
\ket{\psi_8}=\frac{1}{\sqrt{6}}\begin{pmatrix}-1\\0\\1\\-1\\1\\0\\1\\1\end{pmatrix}\;.
\end{align}
We then use the POVM $\{\Pi_i=\ket{\psi_i}\bra{\psi_i}\}$. Note that $\Pi_5$ and $\Pi_6$ are null POVMs, in that $\text{tr}\{\mathcal{S}\Pi_5\}=\text{tr}\{\mathcal{S}\Pi_6\}=0$. As such they can be combined into one POVM element $\Pi_0$ with no loss of information. It is easily calculated that for this example the trace of the inverse of the classical Fisher information matrix gives $22/27$, in agreement with Eq.~\eqref{eqSLDexD}.

\section{A second example supporting result A}
\label{egAapen}
In examples where theorem~\ref{theoremkernel1} cannot be applied, it is not so easy to check when the extended SLD operators can be chosen to commute. In this appendix, we provide a second example, that supports result A for a system with $d=5$, $r=3$. Consider the model with
  \begin{align}
    S^\text{s}=\frac{1}{3}\begin{pmatrix}1&0&0\\0&1&0\\0&0&1\end{pmatrix} \;,
  \end{align}
and the following derivatives 
  \begin{align}
  S_1
  = \frac{1}{3}\begin{pmatrix}\begin{matrix}1&0&0\\0&0&0\\0&0&-1\end{matrix}&\rvline&\begin{matrix}\frac{1}{2}&0\\0&0\\0&\frac{1}{2}\end{matrix}\\\hline
    \begin{matrix}\frac{1}{2}&0&0\\0&0&\frac{1}{2}\end{matrix}&\rvline&\bigzero
   \end{pmatrix}\;,\quad\text{and}\quad
   S_2
  =\frac{1}{3} \begin{pmatrix}\begin{matrix}0&1&0\\1&0&1\\0&1&0\end{matrix}&\rvline&\begin{matrix}0&\frac{1}{2}\\-\frac{1}{2}&\frac{1}{2}\\\frac{1}{2}&0\end{matrix}\\\hline
    \begin{matrix}0&-\frac{1}{2}&\frac{1}{2}\\\frac{1}{2}&\frac{1}{2}&0\end{matrix}&\rvline&\bigzero
    \end{pmatrix}\;.
\end{align}
This model gives rise to the following SLD operators
  \begin{align}
 L_1
  = \begin{pmatrix}\begin{matrix}1&0&0\\0&0&0\\0&0&-1\end{matrix}&\rvline&\begin{matrix}1&0\\0&0\\0&1\end{matrix}\\\hline
    \begin{matrix}1&0&0\\0&0&1\end{matrix}&\rvline&\bigzero\end{pmatrix}\;,\quad\text{and}\quad
  L_2
  = \begin{pmatrix}\begin{matrix}0&1&0\\1&0&1\\0&1&0\end{matrix}&\rvline&\begin{matrix}0&1\\-1&1\\1&0\end{matrix}\\\hline
    \begin{matrix}0&-1&1\\1&1&0\end{matrix}&\rvline&\bigzero\\
    \end{pmatrix}.
\end{align}
We see that
\begin{align}
[L_1,L_2]=\begin{pmatrix}\bigzero&\rvline&\begin{matrix}0&1\\-1&-1\\0&-1\end{matrix}\\\hline
    \begin{matrix}0&1&1\\-1&1&0\end{matrix}&\rvline&\bigzero\end{pmatrix}
\end{align}

We now wish to know if we can choose the SLD operators in the extended space such that $[\mathcal{L}_1,\mathcal{L}_2]=0$. The SLD operators in the extended space can be written
  \begin{align}
  \mathcal{L}_1
  = \begin{pmatrix}\begin{matrix}1&0&0\\0&0&0\\0&0&-1\end{matrix}&\rvline&\begin{matrix}1&0\\0&0\\0&1\end{matrix}&\rvline&\bigzero\\\hline
    \begin{matrix}1&0&0\\0&0&1\end{matrix}&\rvline&L_1^\text{k}&\rvline&L_1^\text{ke}\\\hline
    \bigzero&\rvline&L_1^\text{ek}&\rvline&L_1^\text{e}\end{pmatrix}\;,\quad\text{and}\quad
  \mathcal{L}_2
  = \begin{pmatrix}\begin{matrix}0&1&0\\1&0&1\\0&1&0\end{matrix}&\rvline&\begin{matrix}0&1\\-1&1\\1&0\end{matrix}&\rvline&\bigzero\\\hline
    \begin{matrix}0&-1&1\\1&1&0\end{matrix}&\rvline&L_2^\text{k}&\rvline&L_2^\text{ke}\\\hline
    \bigzero&\rvline&L_2^\text{ek}&\rvline&L_2^\text{e}\end{pmatrix}.
\end{align}
By construction,
$\bra{s_n}[\mathcal{L}_1,\mathcal{L}_2]\ket{s_m}=0$. We are going to
show that in this example, it is not possible to find an extended
space to make  $[\mathcal{L}_1,\mathcal{L}_2]=0$.

 A necessary condition for $[\mathcal{L}_1,\mathcal{L}_2]=0$ in the whole
(support--kernel--extended) space, we need
$\bra{s_n}[\mathcal{L}_1,\mathcal{L}_2]\ket{k_m}=0$ in the
support--kernel off-diagonal subspace. In this subspace,
\begin{align}
  \bra{s_n}[\mathcal{L}_1,\mathcal{L}_2]\ket{k_m}&=0\\
  \implies L_1^\text{s}  L_2^\text{sk} - L_2^\text{s}
  L_1^\text{sk}+ L_1^\text{sk}  L_2^\text{k} - L_2^\text{sk}
  L_1^\text{k}&=0\,,
\end{align}
which does not depend on entries in the extended space.

 In our example, this leads to the following condition
  \begin{align}
  \bra{s_n}[\mathcal{L}_1,\mathcal{L}_2]\ket{k_m} =   \begin{pmatrix}0&1\\-1&-1\\-1&0\end{pmatrix}+\begin{pmatrix}1&0\\0&0\\0&1\end{pmatrix}L_2^\text{k}- \begin{pmatrix}0&1\\-1&1\\1&0\end{pmatrix}L_1^\text{k}=0\,.
  \end{align}
Writing
$L_j^\text{k}=\begin{pmatrix}a_j&b_j+\mathrm{i}c_j\\b_j-\mathrm{i}c_j&d_j\end{pmatrix}$
with all real variables, and working this out explicitly, we get the six conditions
\begin{align}
  \bra{s_n}[\mathcal{L}_1,\mathcal{L}_2]\ket{k_m}=  \begin{pmatrix}  a_2-b_1+\mathrm{i} c_1& 1+b_2+\mathrm{i} c_2 -d_1\\
    -1+a_1-b_1+\mathrm{i}c_1 & -1 + b_1 +\mathrm{i} c_1-d_1\\
    -1+b_2-\mathrm{i} c_2-a_1 &d_2-b_1-\mathrm{i} c_1
    \end{pmatrix}=0\,.
\end{align}
 $\bra{s_n}[\mathcal{L}_1,\mathcal{L}_2]\ket{k_m}=0$ requires that the real and imaginary parts are separately zero. In order for the imaginary parts to be equal to zero we require $c_1=c_2=0$. We now show that there's no solution to the set of equations:
  \begin{align}
    \bra{s_1}[\mathcal{L}_1,\mathcal{L}_2]\ket{k_2}=1+b_2-d_1&=0\\
    \bra{s_2}[\mathcal{L}_1,\mathcal{L}_2]\ket{k_2}=1-b_1+d_1&=0\\
    \bra{s_2}[\mathcal{L}_1,\mathcal{L}_2]\ket{k_1}=1-a_1+b_1&=0\\
    \bra{s_3}[\mathcal{L}_1,\mathcal{L}_2]\ket{k_1}=1+a_1-b_2&=0
  \end{align}
The first two equations imply that $b_1-b_2=2$ but the last two
equations require $b_1-b_2=-2$ which is a contradiction. Therefore, there is no way to choose $\bra{s_n}[\mathcal{L}_1,\mathcal{L}_2]\ket{k_m}=0$.

\bibliographystyle{unsrt}
\bibliography{sld_extension_bib}
\end{document}